\documentclass[12pt, a4paper]{article}
\usepackage{latexsym,amsmath,amsfonts,amssymb,amsthm}
\usepackage{graphicx}
\usepackage{indentfirst}
\usepackage{cite}
\usepackage[cp1251]{inputenc}

\newtheorem{theorem}{Theorem}

\newtheorem{remark}{Remark}

\begin{document}

\begin{center}
 {\Large \bf Lie symmetries of fundamental solutions of one (2+1)-dimensional ultra-parabolic Fokker--Planck--Kolmogorov equation}

\medskip \medskip

{\bf Sergii Kovalenko, $^{1}$} {\bf Valeriy Stogniy $^2$} {\bf and  Maksym Tertychnyi $^3$}
\medskip \medskip \medskip
\\
{\it $^1$~Department of Physics, Faculty of Oil, Gas and Nature Engineering,
\\
Poltava National Technical University, Poltava, Ukraine}
\\
(e-mail: kovalenko@imath.kiev.ua)
\medskip
\\
{\it $^2$~Department of Mathematical Physics, Faculty of Physics and Mathematics,
\\
National Technical University of Ukraine "Kyiv Polytechnic Institute", Kyiv, Ukraine}
\\
(e-mail: valeriy\_stogniy@mail.ru)
\medskip
\\
{\it $^3$~Department of Mathematics and Statistics, Faculty of Science,
\\
University of Calgary, Calgary, Canada}
\\
(e-mail: maksym.tertychnyi@gmail.com)

\medskip \medskip

\end{center}

\begin{abstract}
A (2+1)-dimensional linear ultra-parabolic Fokker--Planck--Kolmogorov equation is investigated from the group-theoretical point of view. By using the Berest--Aksenov approach, an algebra of invariance of fundamental solutions of the equation is found. A fundamental solution of the equation under study is computed in an explicit form as a weak invariant solution.
\end{abstract}


\medskip


\newpage

\section{Introduction}

 The idea that a solution of any well-posed boundary problem for a particular linear partial differential equation (PDE) can be reduced to the construction of some solution of a special type (known as a fundamental solution) is one of the most powerful approaches in a classical mathematical physics. The method of integral transformations is a powerful and well-developed tool for the construction of the solutions of such a type.  Unfortunately, this method is efficient only for linear PDE's with constant(or, at least, analytic) coefficients.  Moreover, in some cases it is very difficult to investigate qualitative properties of fundamental solutions, because they can not be presented in an explicit form, but only in the form of inverse integral transformations. As a result, the development of more direct methods for the construction of fundamental solutions is an important problem of a modern mathematical physics. Especially, it is urgent in the case  of PDE's with alternating coefficients. One of such modern methods is a classical Lie method for investigation of symmetrical properties of PDE's.

Group analysis of differential equations is a mathematical theory with an area of interest in the symmetry properties of differential equations.  Basics of group analysis are in fundame\-ntal works of   S. Lie and his scholars (see, for instance, \cite{Lie}). Namely, Lie developed and was the first to use the tools of symmetry reduction.  The main idea  of this method is to search for the solution of the equation under consideration in the form of a special substitution (ansatz), that reduces the given equation to the differential equation with less number of independent variables.

Further development of group-theoretical methods of differential equations is mainly connected with works  of the following mathematicians: G.~Bikhoff \cite{birk}, L.\,I.~Sedov \cite{sedov}, A.\,J.\,A.~Mo\-rgan \cite{morg}, L.\,V.~Ovsiannikov \cite{ovs}, N.\,H.~Ibragimov \cite{ibr85,ibr92, grig10}, P.\,J.~Olver \cite{olv86, olv95}, W.\,I.~Fushchich \cite{fn83,fn90,fss} and others. At present, symmetry properties of many well-known equations of mechanics, gas dynamics, quantum physics, etc were  investigated.  Detailed analysis of results of research of symmetry properties of a wide range of linear and non-linear differential equations can be found in the monographs  \cite{ovs,ibr85,fn83,fn90,fss,andr,miler,lahno}.

It is well-known, that as a rule fundamental solutions of linear PDE's are invariant with respect to the transformations, admitted by the prescribed equation \cite{ibr92,aks95,aks09,berest91,berest93,berest95,gaz,crad04,crad09,crad12}. In particular, fundamental solutions of  classical equations of mathematical physics such as the Laplace equation, the Heat equation, the Wave equation, etc have this property  (see, for example, \cite{aks09}). It means, that if a linear PDE (especially in case of an equation with alternating coefficients) has non-trivial symmetry properties, then we can use group-theoretical methods to construct its fundamental solution.

The object of our investigation is the linear (2+1)-dimensional ultra-parabolic Fokker--Planck--Kolmogorov equation
\begin{equation}\label{0.1}
u_t - u_{xx} + x u_y = 0, \quad (x,y,t) \in \mathbb{R}^3,
\end{equation}
where $u = u(t,x,y)$, $u_t = \frac{\partial u}{\partial t}$, $u_y = \frac{\partial u}{\partial y}$, $u_{xx} = \frac{\partial^2 u}{\partial x^2}$.

Eq. \eqref{0.1} is the simplest and lowest dimensional version of the following (2$n$+1)-dimensional linear ultra-parabolic equation
\begin{equation}\label{0.2}
u_t - \sum_{j, \, l=1}^{n}(k_{j\,l}(t,x,y) u)_{x_j x_l} + \sum_{j=1}^{n}(f_j(t,x,y) u)_{x_j} + \sum_{j=1}^{n} x_j u_{y_j} = 0, \quad (x,y,t) \in \mathbb{R}^{2n+1}.
\end{equation}

A.\,N. Kolmogorov introduced \eqref{0.2} in 1934 to describe the probability density of a system with $2n$ degrees of freedom \cite{kolm}. Here, the $2n$-dimensional space is the phase space, $x = (x_1, \ldots, x_n)$ is the velocity and $y = (y_1, \ldots, y_n)$ is the position of the system. It should be also stressed that Eq. \eqref{0.1} arises in mathematical finance in some generalization of the Black--Scholes model (see, for instance, \cite{bar01, pas05}).

Research of symmetry properties of Eq.  \eqref{0.1} was commenced in work\cite{lan94}, where several point transformations of symmetries were found. But complete group analysis for this equation was not performed.

Maximal invariance algebra of Eq. \eqref{0.1} in Lie sense was calculated in work \cite{spichak}. In addition, in this article an optimal system of sub-algebras was found for the calculated invariance algebra.  Using two-dimensional sub-algebras, it was also performed the symmetry reduction and constructed several explicit  invariant solutions of Eq. \eqref{0.1}.

In this article, we continue research of symmetry properties of Eq.  \eqref{0.1}. The goal of our work is to find an invariance algebra of fundamental solutions of this equation and to construct in an explicit form the fundamental solution of Eq. \eqref{0.1} using already found symmetry algebra.

\section{Symmetries of fundamental solutions of  linear PDE's}

Consider linear homogeneous PDE of $p$-th order with  $m$ independent variables
\begin{equation}\label{1.1}
Lu \equiv \sum_{|\alpha|=0}^{p} A_{\alpha}(x) D^{\alpha}u=0, \quad x \in \mathbb{R}^m.
\end{equation}
In \eqref{1.1} we use standard notations: $x = (x^1, \ldots, x^m)$, $\alpha = (\alpha_1, \ldots, \alpha_m)$ is a multi-index with integer non-negative components, $|\alpha| = \alpha_1 + \ldots + \alpha_m$;
\[
D^{\alpha} \equiv \left(\frac{\partial}{\partial x^1} \right)^{\alpha_1} \ldots \left(\frac{\partial}{\partial x^m} \right)^{\alpha_m};
\]
$A_{\alpha}(x)$ are some smooth functions of the variable $x$.

\emph{Fundamental solution} of Eq. \eqref{1.1} is a function $u(x,x_0)$ (namely, generalized), that yields the following equation
\begin{equation}\label{1.2}
Lu = \delta(x-x_0),
\end{equation}
where $\delta(x-x_0)$ is the Dirac delta function.

Standard methods to find fundamental solutions of linear PDE's are the method of integral transformations (especially, in case of the equations with constant coefficients), the Green's functions method, etc. \cite{vlad, evans}. Here we consider an algorithm to find fundamental solutions of linear homogeneous PDE's by using symmetry groups of this equation.

Remind, that a non-degenerate local substitution of the variables  $x,u$
\begin{equation}\label{1.3}
\bar{x}^i = f^i(x,u,a), \ \bar{u} = g(x,u,a), \quad i = 1, \ldots, m,
\end{equation}
depending on a continuous parameter $a$ is called a \emph{symmetry transformation} of Eq. \eqref{1.1}, if this equation does not change its form with respect to the new variables $\bar{x}$ and $\bar{u}$. A set $G$ of all such transformations forms a Lie group (local, more precisely), that is called a \emph{symmetry group} (or an acceptable group) of Eq. \eqref{1.1}.

According to the Lie theory, the construction of the symmetry group $G$ of Eq. \eqref{1.1} is equivalent to finding its infinitesimal transformations:
\begin{equation}\label{1.4}
\bar{x}^i \approx x^i + a \cdot \xi^i(x,u), \ \ \bar{u} \approx u + a \cdot \eta(x,u), \quad i = 1, \ldots, m,
\end{equation}
where  $x^i(x,u)$ and $\eta(x,u)$ are some smooth functions.

Linear differential operator of the first order
\begin{equation}\label{1.5}
X = \sum_{i=1}^{m} \xi^i(x,u) \frac{\partial}{\partial x^i} + \eta(x,u) \frac{\partial}{\partial u}
\end{equation}
is called an \emph{infinitesimal operator} of the group $G$. The operator $X$ is also called a symmetry operator  of Eq. \eqref{1.1}.

The group transformations \eqref{1.3} corresponding  to the infinitesimal transformations \eqref{1.4} with the operator \eqref{1.5} are found using so called Lie equations:
\[
\frac{d \bar{x}^i}{d a} = \xi^i(\bar{x},\bar{u}), \ \frac{d \bar{u}}{d a} = \eta(\bar{x},\bar{u}), \quad i = 1, \ldots, m
\]
with the initial conditions
\[
\bar{x}^i \vert_{a=0} = x^i, \ \bar{u} \vert_{a=0} = u.
\]

So called infinitesimal criterion of invariance plays a  fundamental role in the symmetry analysis of differential equation, that is in case of Eq.  \eqref{1.1} can be formulated in the following form.

\begin{theorem}\label{t1}
The infinitesimal operator \eqref{1.5}  is a symmetry operator of Eq. \eqref{1.1}, if and only if there exists such a function $\lambda = \lambda(x)$, that yields the following identities:
\begin{equation}\label{1.6}
\underset{p}{X}(Lu) \equiv \lambda(x) \cdot Lu
\end{equation}
for any function $u = u(x)$ from the domain of Eq. \eqref{1.1}.
\end{theorem}

In Eq. \eqref{1.6}, $\underset{p}{X}$ is the prolongation of $p$-th order of the infinitesimal operator \eqref{1.5}, that is calculated by the well-known formula \cite{ovs,fss,lahno,olv86}:
\[
\underset{p}{X} = X + \sum_{k=1}^{p} \sum_{i_1, \ldots, i_k = 1}^{m} \varphi^{i_1 \cdots \, i_k} \frac{\partial}{\partial u_{i_1 \cdots \, i_k}},
\]
where
\[
\varphi^{i_1 \cdots \, i_k} = D_{i_1} \cdots \, D_{i_k} \left(\eta - \sum_{i=1}^m \xi^i u_i \right) + \sum_{i=1}^m \xi^i u_{i_1 \cdots \, i_k i},
\]
\[
i_1, \ldots, i_k = 1, \ldots, m, \ \ k = 1, \ldots, p.
\]
Here we denote by $D_i$ the operator of total differentiation with respect to the variable  $x_i$:
\[
D_i = \frac{\partial}{\partial x^i} + u_i \frac{\partial}{\partial u} + \sum_{j=1}^{\infty} \sum_{i_1, \ldots, i_j = 1}^{m} u_{i_1 \cdots \, i_j i} \frac{\partial}{\partial u_{i_1 \cdots \, i_j}}, \quad i = 1, \ldots, m.
\]

It was shown in  \cite{bluman}, that the linear homogeneous PDE \eqref{1.1} given additional conditions  $p \geq 2$ and $m \geq 2$ admits symmetry operators only of such a form
\begin{equation}\label{1.7}
X = \sum_{i=1}^m \xi^i(x) \frac{\partial}{\partial x^i} + \left(\alpha(x) u + \beta(x)\right) \frac{\partial}{\partial u}.
\end{equation}
In the class of infinitesimal operators of the form \eqref{1.7}, the maximal invariance algebra of Eq. \eqref{1.1} as a vector space is a direct sum of two sub-algebras: the sub-algebra, that consists of the operators of the form
\begin{equation}\label{1.8}
X = \sum_{i=1}^m \xi^i(x) \frac{\partial}{\partial x^i} + \alpha(x) u \frac{\partial}{\partial u},
\end{equation}
and the infinite-dimensional subalgebra, that is generated by the operators
\begin{equation}\label{1.9}
X = \beta(x) \frac{\partial}{\partial u},
\end{equation}
where $\beta = \beta(x)$ is an arbitrary smooth solution of Eq. \eqref{1.1}.

It is clear, that the infinitesimal operators \eqref{1.9} are symmetry operators of Eq. \eqref{1.2}. Hence, in the sequel we will consider only operators of the form given by  \eqref{1.8}.

Constructive method to find symmetries of the form \eqref{1.8} of linear inhomogeneous  PDE's with $\delta$-function  in a right-hand side was proposed in works \cite{aks95,berest93}. In the same articles, it was introduced an algorithm for construction of invariant fundamental solutions of the equations of the form  \eqref{1.1}.

Main result of these works is in the following statement.

\begin{theorem}\label{t2}
The Lie algebra of symmetry operators of the form  \eqref{1.8} of Eq. \eqref{1.2} is a sub-algebra of the Lie algebra of symmetry operators of Eq. \eqref{1.1}, which is defined by the following conditions:
\begin{equation}\label{1.10}
\xi^i(x_0) = 0,
\end{equation}
\begin{equation}\label{1.11}
\lambda(x_0) + \sum_{i=1}^m \frac{\partial \xi^i(x_0)}{\partial x^i} = 0,
\end{equation}
where  $i = 1, \ldots, m$.
\end{theorem}

Formulate the algorithm to find fundamental solutions of a linear PDE using properties of its invariance algebra:
\begin{itemize}
    \item[1.] find a general form of symmetry operator of Eq.  \eqref{1.1} and a relevant function $\lambda(x)$, that yields  identities \eqref{1.6};
    \item[2.] using conditions \eqref{1.10} and  \eqref{1.11}, find a  Lie algebra of symmetry operators of Eq.  \eqref{1.2};
    \item[3.] construct invariant fundamental solutions of Eq.  \eqref{1.1} using symmetry operators of Eq.~\eqref{1.2}.
\end{itemize}

\begin{remark}
Formulated algorithm for construction of fundamental solutions using symmetries of linear PDE's with $\delta$-function  in a right-hand side is especially efficient for multi-dimensional linear equations and for the equations with alternating coefficients in  case, if they allow rather wide invariance algebras.
\end{remark}

\begin{remark}
To find generalized invariant fundamental solutions, it is necessary to solve reduced equations  (these equations are written in invariants of corresponding transformation groups) on the set of generalized functions (see, for example, \cite{berest91,berest93,dap02}).
\end{remark}

\section{Symmetries of fundamental solutions of Eq. \eqref{0.1}}

Apply the method  described in the previous section to the object of our research, in other words to the linear Fokker--Planck--Kolmogorov equation \eqref{0.1}.

We define  \emph{fundamental solution } of Eq. \eqref{0.1} as a  generalized function $u=u(t,x,y,t_0,x_0,y_0)$, that depends on  $t_0,x_0,y_0$  as parameters and yields the equation
\begin{equation}\label{2.1}
u_t-u_{xx}+xu_y = \delta(t-t_0,x-x_0,y-y_0),
\end{equation}
given the additional condition $u\vert_{t<t_0}=0$.

In work  \cite{spichak}, it was shown that the maximal invariance algebra of Eq.  \eqref{0.1} is generated by the following infinitesimal operators:
\[
X_1 = \partial_x+t\partial_y, \ X_2 = 2t\partial_t+x\partial_x+3y\partial_y-2u\partial_u,
\]
\[
X_3 = t^2\partial_t+(tx+3y)\partial_x+3ty\partial_y-(2t+x^2)u\partial_u,
\]
\[
X_4 = 3t^2\partial_x+t^3\partial_y+3(y-tx)u\partial_u, \ X_5 = 2t\partial_x+t^2\partial_y-xu\partial_u,
\]
\[
X_6 = \partial_t, \ X_7 = \partial_y, \ X_8 = u\partial_u, \ X_{\infty} = \beta(t,x,y)\partial_u.
\]
In the last operator, the function  $\beta = \beta(t,x,y)$ is an arbitrary smooth solution of the equation under study.

\begin{theorem}\label{t3}
Eq. \eqref{2.1} admits an infinite-dimensional Lie algebra of symmetry operators with the following basis of finite-dimensional part:
\[
\begin{split}
& Y_1 = 2(t-t_0)\partial_t+(x-x_0)\partial_x-(x_0(t-t_0)-3(y-y_0))\partial_y-4u\partial_u,\\
& Y_2 = (t^2-t_0^2)\partial_t+((tx+3y)-(t_0x_0+3y_0))\partial_x+(3(y-y_0)t-t_0x_0(t-t_0))\partial_y-\\
& \qquad -(2(t-t_0)+x^2-x_0^2)u\partial_u,\\
& Y_3 = 3(t^2-t_0^2)\partial_x+(t^3-3t_0^2t+2t_0^3)\partial_y-3(tx-y-(t_0x_0-y_0))u\partial_u,\\
& Y_4 = 2(t-t_0)\partial_x+(t-t_0)^2\partial_y-(x-x_0)u\partial_u.
\end{split}
\]
\end{theorem}

\begin{proof}
Write the general form of symmetry operator of the finite-dimensional part of invariance algebra of Eq. \eqref{0.1}
\begin{equation}\label{2.2}
X=\sum_{i=1}^8 a_iX_i,
\end{equation}
or in a more detailed form
\[
\begin{split}
X = & (2a_2t+a_3t^2+a_6)\partial_t+(a_1+a_2x+a_3(tx+3y)+3a_4t^2+2a_5t)\partial_x+\\
& + (a_1t+3a_2y+3a_3ty+a_4t^3+a_5t^2+a_7)\partial_y+\\
& \quad + (-2a_2-a_3(2t+x^2)+3a_4(y-tx)-a_5x+a_8)u\partial_u,
\end{split}
\]
where $a_i \, (i=1,\ldots,8)$ are any real constants.

Substituting the infinitesimal operator \eqref{2.2} into Eq. \eqref{1.6}, where we put $Lu \equiv u_t-u_{xx}+xu_y$, $p=2$, we find the function  $\lambda = \lambda(t,x,y)$, that corresponds to this operator:
\begin{equation}\label{2.3}
\lambda(t,x,y) = -4a_2-a_3(4t+x^2)+3a_4(y-tx)-a_5x+a_8.
\end{equation}

Substituting \eqref{2.2} and \eqref{2.3} into Eqs. \eqref{1.10} and \eqref{1.11} (see, Th. \ref{t1}), we obtain the following equalities:
\[
\begin{split}
& 2a_2t_0+a_3t_0^2+a_6 = 0;\\
& a_1+a_2x_0+a_3(t_0x_0+3y_0)+3a_4t_0^2+2a_5t_0 = 0;\\
& a_1t_0+3a_2y_0+3a_3t_0y_0+a_4t_0^3+a_5t_0^2+a_7 = 0;\\
& 2a_2-a_3(2t_0-x_0^2)+3a_4(y_0-t_0x_0)-a_5x_0+a_8=0,
\end{split}
\]
from which we easily derive
\[
\begin{split}
& a_1 = -a_2x_0-a_3(t_0x_0+3y_0)-3a_4t_0^2-2a_5t_0;\\
& a_6 = -2a_2t_0-a_3t_0^2;\\
& a_7 = a_2(t_0x_0-3y_0)+a_3t_0^2x_0+2a_4t_0^3+a_5t_0^2;\\
& a_8 = -2a_2+a_3(2t_0-x_0^2)-3a_4(y_0-t_0x_0)+a_5x_0.
\end{split}
\]

Substituting in  \eqref{2.2} the constants $a_1,a_6,a_7$, and $a_8$ by the calculated expressions and having split with respect to the independent constants  $a_2,a_3,a_4,a_5$, we obtain that the finite-dimensional part of invariance algebra of Eq. \eqref{2.1} is four-dimensional and generated by the following operators:
\[
\begin{split}
& Y_1 = 2(t-t_0)\partial_t+(x-x_0)\partial_x-(x_0(t-t_0)-3(y-y_0))\partial_y-4u\partial_u,\\
& Y_2 = (t^2-t_0^2)\partial_t+((tx+3y)-(t_0x_0+3y_0))\partial_x+(3(y-y_0)t-t_0x_0(t-t_0))\partial_y-\\
& \qquad -(2(t-t_0)+x^2-x_0^2)u\partial_u,\\
& Y_3 = 3(t^2-t_0^2)\partial_x+(t^3-3t_0^2t+2t_0^3)\partial_y-3(tx-y-(t_0x_0-y_0))u\partial_u,\\
& Y_4 = 2(t-t_0)\partial_x+(t-t_0)^2\partial_y-(x-x_0)u\partial_u.
\end{split}
\]

Invariance property of Eq. \eqref{2.1} with respect to the operators of the form $\beta(t,x,y)\partial_u$, where  $\beta(t,x,y)$ is an arbitrary smooth solution of Eq. \eqref{0.1} is straightforward.

The proof is now completed.
\end{proof}

Show how the results of Th. \ref{t3} can be applied to the construction of invariant fundamental solutions of Eq. \eqref{0.1}. Use for this, for example, the two-dimensional algebra of operators $\langle Y_1, Y_4 \rangle$. Classical invariants corresponding to these operators can be derived from the system of equations
\[
Y_1 I = 0, \ Y_4 I = 0,
\]
where $I = I(t,x,y,u)$, and the equations are solved in a classical sense. We obtain after related calculations:
\[
I_1 = (t-t_0)^2 \exp\left[\frac{(x-x_0)^2}{4(t-t_0)}\right]u, \ I_2 = \frac{(t-t_0)(x+x_0) - 2(y-y_0)}{(t-t_0)^{3/2}}.
\]
 Classical invariant solution is found from the equality $I_1 = \varphi(I_2)$, from which we obtain a substitution (ansatz)
\begin{equation}\label{2.4}
u = \frac{1}{(t-t_0)^2}\exp\left[- \frac{(x-x_0)^2}{4(t-t_0)}\right] \varphi(\omega), \ \omega = I_2,
\end{equation}
that reduces the equation under study \eqref{0.1} to the following ordinary differential equation:
\[
\frac{d^2\varphi}{d\omega^2}+\frac{3}{2}\, \omega \frac{d\varphi}{d\omega}+\frac{3}{2} \, \varphi = 0.
\]
It is well-known, that a particular solution of this equation is the following function \cite[P. 216]{pol}:
\begin{equation}\label{2.5}
\varphi = C \exp\left[-\frac{3}{4}\,\omega^2 \right].
\end{equation}
Substituting \eqref{2.5} to \eqref{2.4}, we derive a classical invariant solution of Eq.  \eqref{0.1} corresponding to the operators  $Y_1$ and $Y_4$:
\begin{equation}\label{2.6}
u = \frac{C}{(t-t_0)^2} \, \exp\left[-\frac{(x-x_0)^2}{4(t-t_0)}-\frac{3}{(t-t_0)^3} \left(y-y_0 -(t-t_0) \frac{x+x_0}{2} \right)^2 \right].
\end{equation}

Substituting \eqref{2.6} to Eq. \eqref{2.1}, it is easy to show that  $Lu(t,x,y)=0$. As a result, the solution  \eqref{2.6} is not a fundamental solution of the Fokker--Planck--Kolmogorov equation \eqref{0.1}. To construct its weak invariant solution we use Statement 1 from work   \cite{berest93a}, in other words we search for  weak invariant solutions in the form
\[
u = \frac{h(t,x,y)}{(t-t_0)^2} \, \exp\left[-\frac{(x-x_0)^2}{4(t-t_0)}-\frac{3}{(t-t_0)^3} \left(y-y_0 -(t-t_0) \frac{x+x_0}{2} \right)^2 \right],
\]
 where $h = h(t,x,y) \in \mathcal{D}'(\mathbb{R}^3)$ (denote by  $\mathcal{D}'= \mathcal{D}'(\mathbb{R}^3)$  the space of generalized functions). The equations $Y_1 u = 0$ and $Y_4 u = 0$ give us correspondingly:
\[
\begin{split}
& 2(t-t_0)h_t+(x-x_0)h_x-(x_0(t-t_0)-3(y-y_0))h_y=0,\\
& 2(t-t_0)h_x+(t-t_0)^2h_y=0.
\end{split}
\]
It is easy to see that the generalized function  $h(t,x,y) = C_1 \theta(t-t_0) + C_0$ is a solution of these equations (here, $ \theta(t-t_0)$ is the Heaviside step function). We obtain using this fact:
\begin{equation}\label{2.7}
u = \frac{C_1 \theta(t-t_0) + C_0}{(t-t_0)^2} \, \exp\left[-\frac{(x-x_0)^2}{4(t-t_0)}-\frac{3}{(t-t_0)^3} \left(y-y_0 -(t-t_0) \frac{x+x_0}{2} \right)^2 \right].
\end{equation}
Substituting  \eqref{2.7} to Eq. \eqref{2.1}, we find the constant  $C_1 = \frac{\sqrt 3}{2 \pi}$.

Hence, the fundamental solution of the Fokker--Planck--Kolmogorov equation \eqref{0.1}
\begin{equation}\label{2.8}
u = \frac{\frac{\sqrt 3}{2 \pi} \, \theta(t-t_0) + C_0}{(t-t_0)^2} \, \exp\left[-\frac{(x-x_0)^2}{4(t-t_0)}-\frac{3}{(t-t_0)^3} \left(y-y_0 -(t-t_0) \frac{x+x_0}{2} \right)^2 \right]
\end{equation}
is found as a weak invariant solution with respect to a two-dimensional algebra  $\langle Y_1, Y_4 \rangle$ of point symmetries of Eq. \eqref{2.1}. In formula \eqref{2.8} we can put  $C_0 = 0$, because a fundamental solution is defined up to the addition of any solution of a homogeneous equation.



\begin{remark}
The fundamental solution  \eqref{2.8} was found by A.\,N. Kolmogorov \cite{kolm1} without using the methods of symmetry analysis of differential equations. Our calculations give group-theoretical background for this solution and justify once again an empiric observation, that fundamental solutions of linear PDE's should be searched among the invariant solutions.
\end{remark}

\begin{remark}
Using a similar scheme, as it was performed for Eq. \eqref{0.1}, we can investigate the symmetry properties of fundamental solutions of other many-dimensional and more complicated Kolmogorov type equations  \eqref{0.2}. In this article, we restricted our attention to Eq.  \eqref{0.1} in order not to lose   the main idea of Berest--Aksenov method of construction of fundamental solutions of linear homogeneous PDE's using its group properties under complicated technical calculations.
\end{remark}

\section{Conclusions}

In this article, using the Berest--Aksenov method \cite{aks95,berest93} we found an invariance algebra of fundamental solutions of linear Fokker--Planck--Kolmogorov equation  \eqref{0.1}. Its operators were used to construct invariant fundamental solutions of this equation. It was shown that the fundamental solution  \eqref{2.8} of Eq. \eqref{0.1}, found by A.\,N. Kolmogorov is a weak invariant fundamental solution.

\end{document}